\documentclass[10pt,aps,superscriptaddress,twocolumn,twoside,floatfix,pra,nofootinbib,a4paper]{revtex4-2}

\usepackage{enumerate,appendix}
\usepackage{braket}
\usepackage{amsmath}
\usepackage{amsthm}
\usepackage{amssymb}
\usepackage{amsfonts}
\usepackage{dsfont}
\usepackage{latexsym}
\usepackage{yhmath}
\usepackage{mathtools}
\usepackage{mathrsfs}
\usepackage{slashed}
\usepackage{cancel}
\usepackage[normalem]{ulem} 
\usepackage{bm} 
\usepackage{bbold}
\usepackage{physics}
\usepackage{graphicx, rotating}
\usepackage{epstopdf}
\usepackage{epsfig}
\usepackage{graphics,subfigure}
\usepackage{comment}
\usepackage[utf8]{inputenc}
\usepackage{soul}
\usepackage{multirow}
\usepackage{color}
\usepackage[dvipsnames]{xcolor}
\usepackage{hyperref}
\def \beq {\begin{equation}}
\def \eeq {\end{equation}}
\def \ba {\begin{array}}
\def \ea {\end{array}}
\def\bea{\begin{eqnarray}}
\def\eea{\end{eqnarray}}

\def \dis {\displaystyle}

\newcommand{\AB}[1]{\textcolor{Purple}{{\bf #1}}}

\newtheorem{theorem}{Theorem}[section]

\newtheorem{lemma}[theorem]{Lemma}
\newtheorem{proposition}{Proposition}[section]

\begin{document}

\title{Maximal CHSH violation for qubit-qudit states }

\author{Alexander Bernal}\email{alexander.bernal@csic.es}\affiliation{Instituto de Física Teórica, IFT-UAM/CSIC, Universidad Autónoma de Madrid, Cantoblanco, 28049 Madrid, Spain}
\author{J. Alberto Casas}\email{j.alberto.casas@gmail.com}\affiliation{Instituto de Física Teórica, IFT-UAM/CSIC, Universidad Autónoma de Madrid, Cantoblanco, 28049 Madrid, Spain}
\author{Jes\'us~M. Moreno}\affiliation{Instituto de Física Teórica, IFT-UAM/CSIC, Universidad Autónoma de Madrid, Cantoblanco, 28049 Madrid, Spain}

\begin{abstract}
We evaluate the maximal CHSH  violation for a generic  (typically mixed) qubit-qudit state, 
obtaining easily computable expressions in arbitrary qudit dimension. This represents the optimal (2-2-2) Bell nonlocality for this kind of systems.
The work generalizes the well-known Horodeckis's result for a qubit-qubit  setup. We also give simple lower and upper bounds on that violation. We apply our general results to address a number of issues. Namely, we obtain a bound on the degree of purity required in a system to exhibit nonlocality and study the statistics of nonlocality in random density matrices. Besides, we show the impossibility of improving the amount of CHSH-violation by embedding the qudit in a Hilbert space of larger dimension. Finally, the general result is illustrated with a family of density matrices in the context of a qubit-qutrit system. 
\end{abstract}

\maketitle


\tableofcontents

\section{Introduction}
\label{sec:Intro}

Violation of Bell-like inequalities represents a crucial test of the character of the fundamental laws of nature, as it is incompatible with local-realism and in particular with local hidden-variables theories. The most popular variant of these inequalities is the CHSH version \cite{Clauser:1969ny}. 
Namely, given a bipartite system where each party, Alice and Bob, performs two binary-outcome measures (usually referred to as a (2-2-2) setup \cite{Brunner:2013est}), the joined probability distribution of the outcomes given the inputs is compatible with local realism if and only if
\begin{equation}
    E (0,0)+ E (0,1)+E (1,0)-E (1,1) \leq 2,
\end{equation}
where we denote $E(x,y)= P(1,1|x,y)-P(1,-1|x,y)-P(-1,1|x,y)+P(-1,-1|x,y)$ and $P(a,b|x,y)$ is the probability of getting outcomes $a,b=\pm1$ given inputs $x,y=0,1$. This statement over probabilities is equivalent, in the operator picture, to $\langle {\cal O}_{\rm Bell}\rangle\leq 2$ with
${\cal O}_{\rm Bell}=A \otimes (B+B') + A'\otimes (B-B')$ and $A,A',B,B'$ Hermitian operators with eigenvalues contained in $[-1,1]$. As it is well known, quantum mechanics can violate this CHSH inequality for certain entangled states. More precisely, if the state is pure there is always a choice of $A, A', B, B'$ which violates CHSH \cite{GISIN1991201}. If it is a mixture, that is not guaranteed \cite{HORODECKI1995340}.

Of course, given a  quantum  state, the amount of potential Bell violation depends on the choice of the four $A, A', B, B'$ observables. It was shown in ref.\cite{HORODECKI1995340} that for a generic qubit-qubit state, $\rho$, expressed as
\begin{eqnarray}
\rho  =  \frac{1}{4}  &&
\left( \mathbb{1}_2\otimes \mathbb{1}_2 +\sum_i(B_i^+ \sigma_i\otimes \mathbb{1}_2 + B_i^- \mathbb{1}_2\otimes \sigma_i)  + \right. \nonumber \\
 &&  \phantom{xxxx} \left.+ \sum_{ij}  C_{ij} \sigma_i\otimes \sigma_j \right) 
\label{generalrho2qu}
\end{eqnarray}
(with $B_i^\pm, C_{ij}$ real coefficients), the maximum value of $\langle {\cal O}_{\rm Bell}\rangle$ is given by
\begin{eqnarray}\label{CHSH-abC_max} \nonumber
    \expval{{\cal O}_{\rm Bell}}_{\max}
 = &&\max_{{A},{A'},{B},{B'}}\expval{A \otimes (B+B') + A'\otimes (B-B')}\\
=&&\, 2\sqrt{\kappa_1+\kappa_2}\ ,
\end{eqnarray}
where $\kappa_1,\kappa_2$ are the largest eigenvalues of $C^T C$. The authors provided also the explicit choice of $A, A', B, B'$ leading to this maximum value. In this way one can easily check whether a (qubit-qubit) state generates probability distributions incompatible with local realism. Complementary results are derived when only fixing spectral properties of the density matrix $\rho$, such as the eigenvalues or the purity. For these cases, the maximum value $\expval{{\cal O}_{\rm Bell}}_{\max}$ for a given (pure or mixed) state has been analytically obtained \cite{Verstraete_2001,Verstraete_2002,Batle_2011}.

For a general quantum system, beyond the qubit-qubit case, things get much more involved. As a matter of fact, there are not closed expressions for the facets defining the region (polytope) of probability distributions compatible with local realism (also known as tight Bell inequalities). Apart from the CHSH inequalities, whose liftings are always tight \cite{Pironio_2005}, the celebrated CGLMP inequalities \cite{Collins:2002sun} also represent some facets of such polytope, but in general they do not provide a complete description of them. Nevertheless, for Bell scenarios consisting of two parties (Alice, Bob) in which Alice performs two measurements with two outcomes each, it was shown by Pironio \cite{Pironio_2014} that all the facets defining the associated ``classical" polytope are given by CHSH-type inequalities. However, this does not solve the problem of determining the maximum CHSH violation for a given qubit-qudit  state $\rho$, and thus whether the probabilities of physical observables of the system in the CHSH scenario can be described by a classical (local-realistic) theory. 
In particular, Eq.(\ref{CHSH-abC_max}) cannot be extrapolated to higher dimension. Then, in order to get the optimal CHSH violation for a given state of a qubit-qudit system, one should explore all possibilities for the $A, A', B, B'$ observables, a very expensive computational task, as it involves a large number of parameters, which grows rapidly as the dimension of the qudit increases.
On the other hand, qubit-qudit systems are of significant interest for both quantum applications and fundamental studies \cite{PhysRevA.62.032310,PhysRevA.61.062302,PhysRevA.81.062102,Chatterjee:2023}.
Concerning the latter, these systems offer a great theoretical laboratory to explore how quantum properties, such as (Bell) nonlocality, scale beyond the minimal qubit-qubit case, revealing dimensionality-driven effects. 
We give some examples in this paper, e.g. examining the statistics of nonlocality for random  qubit-qudit states; or determining the 
 minimum degree of purity required to present non-local correlations. Also, qubit-qudit systems provide relevant tests of quantum mechanics. For example, there has been a recent effort to certificate the presence of both entanglement and non-locality in high-energy processes at particle colliders \cite{Afik:2025ejh}. 
 In this sense, it remains an open challenge to demonstrate that hybrid states, such as (top quark, $W$ boson) pairs produced at the LHC, exhibit Bell nonlocality. The results of this study facilitate the optimization of this task, as they provide the optimal (2-2-2) test.
 
Qubit-qudit systems are also of interest in quantum information processing, e.g. they enable advanced protocols, such as hybrid quantum error correction, that significantly enhance error correction capabilities \cite{Chizzini:2022zlg}. In general, 
Bell nonlocality in these systems provides a robust method for entanglement certification, which is an additional motivation for the development of optimized nonlocality tests.

In this paper we address the task of evaluating the maximal CHSH violation for a generic qubit-qudit system, obtaining easily computable expressions. 

In section \ref{sec:2} we present our approach to the problem and the general result for maximal CHSH violation in qubit-qudit systems, which is the main result of this paper. 
In section \ref{sec:applications} we apply our general results to address a number of issues. 
In section \ref{sec:3} we give  simple lower and upper bounds on $\expval{{\cal O}_{\rm Bell}}$, which respectively represent sufficient and necessary conditions for violation of local realism. 
In section \ref{sec:purity} we use our general results to obtain a bound on the degree of purity required for nonlocality. Section \ref{sec:random} is devoted to study the statistics of nonlocality in random density matrices.
In section \ref{sec:4} we examine the possibility of improving the amount of Bell-violation by embedding Bob's Hilbert space in one of larger dimension.
In section \ref{sec:5} we illustrate our results by studying a family of density matrices in a qubit-qutrit system. Finally, in section \ref{sec:6} we summarize our results and conclusions.

\section{Main result}\label{sec:2}

Let us consider a general qubit-qudit system, with Hilbert space $ {\cal H}_2 \otimes {\cal H}_d$. Any 
$2d\times 2d$ density matrix in this space can be unambiguously expressed as
\begin{equation}
    \rho=\frac{1}{2}\left[\mathbb{1}_2\otimes\beta_0+\sigma_1\otimes\beta_1+\sigma_2\otimes\beta_2+\sigma_3\otimes\beta_3\right],
    \label{rho}
\end{equation}
where $\sigma_i$ are the standard Pauli matrices
and $\{\beta_0,\vec \beta\}=\{\beta_0,\beta_i\}_i$ are $d\times d$ Hermitian matrices. In particular, the  $\beta_0$  matrix coincides with Bob's reduced density matrix, $\rho_B=\Tr_A \rho=\beta_0$, and therefore verifies  $ \Tr \beta_0  = 1$.\footnote{The rest of the $\beta$  matrices are  easily obtained  by $\beta_i  =  \Tr_A \left(\rho\, (\sigma_i \otimes \mathbb{1}_d   ) \right)$.} Besides, the $\beta-$matrices must  lead to a positive semidefinite $\rho$ matrix; other than that they are arbitrary. The previous expression is a kind of Schmidt decomposition of the $ {\cal H}_2 \otimes {\cal H}_d$ density matrix.
  
Following the result obtained by Pironio \cite{Pironio_2014}, we know that in a bipartite scenario in which one party has binary inputs and outcomes all the facets defining the polytope of Local Hidden Variables (LHV) are given by CHSH-type inequalities . In other words, all the ``tight" Bell-like inequalities (those whose violation is a sufficient and necessary condition to violate local realism in this bipartite Bell scenario) can be written as CHSH-type inequalities. Following the reasoning detailed in Appendix \ref{POVMAppend}, the CHSH inequalities can always be written as $|\langle \mathcal{O}_{\rm Bell}\rangle|\leq 2$, with
\begin{equation}
 {\cal O}_{\rm Bell} = A \otimes (B+B') + A'\otimes (B-B') \ .
 \label{OBell}
 \end{equation}
Here $A$ and $A'$ ($B$ and $B'$) are $2\times2$ ($d\times d$) linear Hermitian observables with eigenvalues $\{+1, -1\} $ ($\{+1, -1\} $ with some degeneracy). This set of eigenvalues arises by considering that each party in the CHSH scenario performs projection-valued measures (PVMs) instead of general positive operator-valued measures (POVMs). Nonetheless, it is known that in binary-outcome scenarios it is sufficient to consider PVMs in order to get the maximal violation \cite{cleve2010consequenceslimitsnonlocalstrategies}. The expectation value of the Bell operator  $\expval{{\cal O}_{\rm Bell}} = \Tr{\rho  {\cal O}_{\rm Bell}}$ reads
\begin{equation}\label{<OB>}
    \begin{aligned}
       \expval{{\cal O}_{\rm Bell}} =&\frac{1}{2}\sum_{i=1}^3 \left(\Tr{\sigma_i A} \, \Tr{\beta_i (B+B')}\right. \\
       &\left.+\Tr{\sigma_i A'} \, \Tr{\beta_i (B-B')} \right).
    \end{aligned}
\end{equation}
As it is well known, for local realistic theories $ \langle {\cal O}_{\rm Bell} \rangle   \leq 2$, while in quantum theories it can reach the Tsirelson's bound $2 \sqrt{2}$ \cite{Cirelson:1980ry}.
Our goal is to find the maximal value:
\begin{equation}
 \mathcal{B}= \max_{A,A',B,B'} \langle {\cal O}_{\rm Bell} \rangle\ .
\label{maxAB}
\end{equation}
For the qubit-qubit case, the cross-terms $\sigma_i\otimes\beta_i$ in (\ref{rho}) can be expressed as $ \frac{1}{2}C_{ij}\sigma_i\otimes\sigma_j$, where $C_{ij}$ is a real matrix. Then, it was shown in \cite{HORODECKI1995340} that the maximum value of $\langle {\cal O}_{\rm Bell} \rangle$ is given by $2\sqrt{\kappa_1+\kappa_2}$, where $\kappa_1,\kappa_2$ are the largest eigenvalues of the $C^TC$ matrix. Such nice result cannot be extrapolated to the qubit-qudit case for various reasons. 
In particular, the freedom for the choice of the $B, B'$ observables is much greater, as they live in the space of $d\times d$ Hermitian matrices.

For this analysis it is convenient to define $\vec{r}_A$, $\vec{r}_{A'}$ and  
$\vec{r}_B$, $\vec{r}_{B'}$  vectors as 
\begin{eqnarray}
 \vec{r}_A  & = & \left(   \Tr{\sigma_1 A}, \Tr{\sigma_2 A},\Tr{\sigma_3 A}  \right ),  \nonumber
 \\
 \vec{r}_B & = & \left(\Tr{\beta_1 B}, \Tr{\beta_2 B},\Tr{\beta_3 B}\right )
 \label{rArB}
\end{eqnarray}
and similar expressions for $\vec{r}_{A'}$ and  $\vec{r}_{B'}$.  Note that these are real vectors from the Hermiticity of the involved matrices. Then Eq.(\ref{<OB>}) reads
\begin{equation}
\langle {\cal O}_{\rm Bell} \rangle = \frac{1}{2} \vec{r}_A (\vec{r}_B+ \vec{r}_{B^\prime}) + 
        \frac{1}{2} \vec{r}_{A^\prime} (\vec{r}_B- \vec{r}_{B^\prime})\ .
 \label{Bell_r}                                                  
\end{equation}
Incidentally,  this expression is explicitly invariant under  simultaneous rotations in the 3-spaces of the $\sigma_i$, $\beta_i$ matrices, which is in turn a consequence of the invariance of $\rho$, Eq.(\ref{rho}), under that operation. 
Now notice that $\frac{1}{2}\vec{r}_A$, $\frac{1}{2}\vec{r}_{A'}$ are unit vectors. This comes from $A, A'$ having eigenvalues $\{+1,-1\}$ and thus vanishing trace\footnote{ We do not consider the case of $A$ or $A'$ proportional to the identity, which leads to no Bell-violation \cite{LANDAU198754}.}, so they can be expressed as
\begin{equation}
A= \sum\frac{1}{2}\Tr{\sigma_i A}\sigma_i= \frac{1}{2}\vec{r}_A \vec{\sigma}
\label{A}
\end{equation}
and similarly for $A'$. 
Now, since $\Tr{A^2}=\Tr{{A'}^2} =\Tr{\mathbb{1}_2}=2$, we get $\|\vec{r}_{A}\|^2=\|\vec{r}_{A^\prime}\|^2=4$. Apart from that, the $\vec{r}_{A}, \vec{r}_{A^\prime}$ vectors are arbitrary since, for any choice of them, the corresponding $A, A'$ observables are given by (\ref{A}). This feature breaks down for higher-dimensional setups. Then, it is necessary to take a different basis on Alice's side for the expansion \eqref{rho} --e.g. the Gell-Mann matrices-- and, in addition, the $A, A'$ operators can be non-traceless. In consequence, there appear constraints not only on the modulus of $\vec r_A, \vec r_{A'}$, but also on their directions. Hence, the whole procedure is no longer practical.

Coming back to the $2\times d$ scenario, for a given pair $(B, B^\prime)$, the optimal choice of  $(A, A^\prime)$ is
$\vec{r}_A   \parallel   ( \vec{r}_B +  \vec{r}_{B^\prime}) $ and 
$ \vec{r}_{A^\prime} \parallel   ( \vec{r}_B -  \vec{r}_{B^\prime})$, so that 
\begin{equation}
\mathcal{B} =\, \max_{B,B'} \Bigl\{
\| \vec{r}_B +  \vec{r}_{B^\prime}\|  +  \| \vec{r}_B-  \vec{r}_{B^\prime}\| \Bigl\}.
 \label{maxBB} 
\end{equation}
As expected,  this expression is also invariant under 3-rotations. Unfortunately, the $\beta-$matrices do not follow the Pauli algebra, so a similar argument cannot be done for the $\vec{r}_{B}, \vec{r}_{B^\prime}$ vectors, in particular they do not have a fixed normalization. As already mentioned, this is part of the extra intricacy of the qubit-qudit case compared to the qubit-qubit one. In order to solve (\ref{maxBB}) it is useful the following lemma:

\begin{lemma}
\label{MaxRotation}
Let $\vec{v} ,\vec{w}$ be two arbitrary vectors. Consider a simultaneous rotation of both vectors within the plane they  span. Let $ \vec{v} (\varphi)$ and $\vec{w} (\varphi)$ be the rotated vectors with $\varphi$ characterizing the rotation angle. Then the following identity holds:
\begin{equation} 
\bigg( \| \vec{v}+  \vec{w}\|  + \| \vec{v} -  \vec{w}\|   \bigg)^2 = 4 \max_{ \varphi }  \bigg\{v_1(\varphi)^2 +w_2(\varphi)^2\bigg\},
 \nonumber
\end{equation}
where the subscripts 1, 2 denote the components of the vectors.
\end{lemma}

\begin{proof} From the sum
\begin{equation}
    \begin{aligned}
        &4(v_1(\varphi)^2 +w_2(\varphi)^2)\\
        &=4(v_1 \cos \varphi + v_2 \sin \varphi )^2+
        4(-w_1 \sin \varphi + w_2 \cos \varphi )^2\\
        &=2 ( v_1^2+v_2^2+w_1^2+w_2^2)\\
        &+2( v_1^2-v_2^2-w_1^2+w_2^2)\cos (2 \varphi )+4 ( v_1 v_2 -w_1 w_2)  \sin (2 \varphi ) 
        \nonumber
    \end{aligned}
\end{equation}
we get that the maximum reads
\begin{equation}
    \begin{aligned}
       &4\max_{ \varphi }\left\{v_1(\varphi)^2 +w_2(\varphi)^2 \right\}\\
       &=2 ( v_1^2+v_2^2+w_1^2+w_2^2)\\
       &+\sqrt{4( v_1^2-v_2^2-w_1^2+w_2^2)^2 +16(v_1 v_2 -w_1 w_2)^2}
       \nonumber
    \end{aligned}
\end{equation}
which coincides with the expanded expression of $\left(\| \vec{v}+  \vec{w}\|  + \| \vec{v} -  \vec{w}\|\right)^2$. 
\end{proof}

The lemma holds when we allow for rotations in 3-space, $\vec{v} ,\vec{w}\rightarrow {\cal R}\vec{v} ,{\cal R}\vec{w}$, {\it i.e.}
\begin{equation}
 \Bigl( \| \vec{v}+  \vec{w}\|  + \| \vec{v} -  \vec{w}\|   \Bigl)^2 = 4\, \max_{ {\cal R} }  \Bigl\{({\cal R}\vec{v})_1^2 +({\cal R}\vec{w})_2^2 \Bigl\},
 \label{lemma3D}
\end{equation}
where ${\cal R}$ is an arbitrary 3-dimensional rotation, characterized by the three Euler angles. This becomes obvious by taking into account that the r.h.s. of this equation reaches its maximum when the two vectors have vanishing third component, $({\cal R}\vec{v})_3 =({\cal R}\vec{w})_3=0$, so that the problem reduces to the above rotation in the plane.

\vspace{0.2cm}
\begin{theorem}
\emph{(Main Theorem)}
\label{Lagrange}
Given a density matrix $\rho$ of a qubit-qudit system, as that in Eq.\eqref{rho}, then the maximum value  $\mathcal{B}=\max_{A,A', B, B'} \Tr{\mathcal{O}_{\rm Bell} \rho}$, with $\mathcal{O}_{\rm Bell}$ given in Eq.\eqref{OBell}, is
\begin{equation}
\boxed{   
\begin{aligned}
\mathcal{B}
&=
\dis{2\, \max_{{\cal R}} 
\sqrt{\| 
({\cal R} \vec{\beta})_1 \|_1 ^2 + \| ({\cal R}\vec{\beta})_2 
\|_1 ^2 } \, } \\
  & =  \dis{2\, \max_{{\cal R}} 
\left[
            \left(\sum_{i=1}^{d} |\lambda_i^{(1)}({\cal R})| \right) ^2  +
           \left( \sum_{i=1}^{d} |\lambda_i^{(2)}({\cal R})| \right) ^2
\right]^{1/2}}\ .
\label{generalresult}
\end{aligned}
}  
\end{equation}
Here, $\lambda_i^{(1,2)} ({\cal R})$ stand for the eigenvalues of the rotated $\beta$ matrices, $({\cal R}\vec{\beta})_1, ({\cal R} \vec{\beta})_2$, and ${\cal R}$ is a general $SO(3)$-rotation.
\end{theorem}

\begin{proof}
Applying Lemma \ref{MaxRotation} to the CHSH expectation value, as given by Eq.(\ref{maxBB}), we get
\begin{equation}
\mathcal{B} = 2 \, \max_{B,B',\, {\cal R}} 
 \sqrt{ \left| ({\cal R}  \vec{r}_B )_1 \right |^2  +  \left| ({\cal R}  \vec{r}_{B^\prime} )_2 \right|^2}\ .
 \label{r_squared}
\end{equation}
From the definition of $\vec{r}_B$, Eq.(\ref{rArB}), $({\cal R} \vec{r}_B )_i= \Tr{({\cal R}\vec{\beta})_i\cdot B}$, and an analogous expression for $({\cal R} \vec{r}_{B'} )_i$ holds. Hence
\begin{equation}
\mathcal{B} =2 \max_{B,B', {\cal R}} 
 \sqrt{ \left| \Tr{({\cal R}  \vec{\beta})_1 \cdot B} \right |^2  +  \left| \Tr{({\cal R}  \vec{\beta})_2 \cdot B'} \right |^2}.
 \label{r_squared2}
\end{equation}
Now we take into account the following \cite{Mirsky1975}: for a generic Hermitian matrix, $M$, with eigenvalues $\lambda_i$, and an arbitrary Hermitian, involutory matrix $B$ (i.e. $B^2=\mathbb{1}_d$), it happens that $\max_{B} \Tr{M\cdot B}=\sum_i |\lambda_i|$.~\footnote{This is called the trace-norm or 1-norm of a matrix, $\|M\|_1=\Tr{\sqrt{M^\dagger M}}=\sum_i |\lambda_i|$, in analogy to the 1-norm of vectors.
}
This maximum is achieved when $B$ is aligned with $M$, i.e. when they are diagonalized by the same unitary matrix, and the signs of the $B$ eigenvalues ($1$ or $-1$) are chosen equal to the signs of the corresponding $\lambda_i$.

Given that this is precisely our case, since $B, B'$ in 
Eq.(\ref{OBell}) are arbitrary Hermitian involutory matrices\footnote{ Here we allow $B, B'$ to be $\pm \mathbb{1}_d$, which is the optimal choice when all $\lambda_i$ have the same sign (we comment below on the meaning of this case). }, then:
\begin{equation}
\mathcal{B}= 2\, \max_{{\cal R}} 
\sqrt{\| ({\cal R} \vec{\beta})_1 \|_1 ^2 + \| ({\cal R}\vec{\beta})_2\|_1 ^2 } \, 
\end{equation}
and the theorem is proven.
\end{proof}

 It follows that a necessary and sufficient condition for CHSH-nonlocality in qubit-qudit states is that the above expression (\ref{generalresult}) is larger than $2$.

Let us briefly comment on some aspects of this result.

\begin{itemize}

\item 
    When all the $\lambda_i^{(1)}({\cal R})$ and/or $\lambda_i^{(2)}({\cal R})$ at the maximum of Eq.(\ref{generalresult}) have the same sign, this entails setting either $B=\pm\mathbb{1}_d$ and/or $B'=\pm\mathbb{1}_d$, which is known to give no violation for CHSH-type inequalities \cite{LANDAU198754}.

\item
To see the computational advantage of the above expression, note the following. In this procedure, given a $\rho$ matrix, once it is expressed in the form (\ref{rho}), we have to perform a (usually numerical) maximization of Eq.(\ref{generalresult}). This implies to scan the three Euler angles of the ${\cal R}$ rotation, which is a very cheap computation. It should be compared with the $4+2d(d-1)$ parameters for each CHSH inequality in the initial expression (\ref{maxAB}). Even in the simplest qubit-qutrit case this represents 16 parameters.

\item The $A, A', B, B'$ observables that realize the maximum Bell-violation are straightforward to obtain. Once we have determined the matrices $({\cal R}\vec{\beta})_1$, $({\cal R}\vec{\beta})_2$ that maximize (\ref{generalresult}) we simply set 
\begin{equation}\hspace{1cm}
    B=U_1 D_1 U_1^\dagger,\quad B'=U_2 D_2 U_2^\dagger,
    \label{BfromBeta}
\end{equation}
where $U_{1,2}$ are unitary matrices diagonalizing $({\cal R}\vec{\beta})_{1,2}$, i.e.
$U_a^\dagger ({\cal R}\vec{\beta})_a U_a={\rm diag}\,(\lambda_i^{(a)})$, and $D_a={\rm diag}({\rm sign}\,[\lambda_i^{(a)}])$. The corresponding $A,A'$ observables are given by Eq.(\ref{A}), with $\vec{r_A}$, $\vec{r_{A'}}$ the unit vectors aligned along 
$( \vec{r}_B +  \vec{r}_{B^\prime})$,~
$ ( \vec{r}_B -  \vec{r}_{B^\prime})$ (see discussion after Eq.(\ref{A})), and 
\begin{equation}\hspace{0.7cm}
    \vec{r}_B  =  \left(\Tr{({\cal R}\vec{\beta})_1 B}, \Tr{({\cal R}\vec{\beta})_2 B}, \Tr{({\cal R}\vec{\beta})_3 B}
\right )
\end{equation}
and similarly for $B'$.  

\item 
Let us finally see that expression (\ref{generalresult}) is consistent with the qubit-qubit result (\ref{CHSH-abC_max}) obtained in ref.\cite{HORODECKI1995340}.

In such scenario, comparing  expressions  (\ref{generalrho2qu}) and (\ref{rho}) for $\rho$, the $\beta$ matrices read $\beta_0=\frac{1}{2}(\mathbb{1}_2+\sum_i B_i^-\sigma_i)$ and $\beta_i=\frac{1}{2}(B_i^+\mathbb{1}_2+\sum_j C_{ij}\sigma_j)$. 
On the other hand, assuming that the state violates a CHSH inequality, the corresponding observables $A, A', B, B'$  must have eigenvalues $\{+1,-1\}$. The other inequivalent possibility, namely one or more observables proportional to the identity, leads to no CHSH-violation  \cite{LANDAU198754}. In that case, the terms involving $B_i^\pm$ are irrelevant as they cancel in $\Tr{\rho {\cal O}_{\rm Bell}}$, Eq. (\ref{<OB>}).
Now, the (real) matrix $C$ can be diagonalized by two orthogonal transformations, ${\cal R}_A, {\cal R}_B\in O(3)$:
\begin{equation}\hspace{1cm}
    C= {\cal R}_A \Sigma {\cal R}_B^T,\quad \Sigma={\rm diag}\{\mu_1, \mu_2, \mu_3\},
\end{equation}
ordered as $\mu_1\geq\mu_2\geq\mu_3\geq0$. This is equivalent to perform appropriate changes of basis in the Alice and Bob Hilbert spaces. Hence, in this new basis
\begin{equation}\hspace{1cm}
\rho=\frac{1}{2}\left(
\mathbb{1}_2\otimes \mathbb{1}_2 
 + \sum_{i} \sigma_i\otimes \beta_i
 + \cdots \right)\ ,\end{equation}
where the dots denote terms which are irrelevant for the previous reasons, and $\beta_i=\frac{1}{2}\mu_i\sigma_i$ up to a sign\footnote{The presence of a negative sign depends on whether or not ${\cal R}_A,{\cal R}_B\in SO(3)$. Nevertheless, this sign is irrelevant for the rest of the reasoning. 
}. Now, from Eq.(\ref{generalresult}) we have to maximize $\| ({\cal R}\vec{\beta})_1 \|_1 ^2 + \| ({\cal R}\vec{\beta})_2 \|_1 ^2$ where ${\cal R}$ is an arbitrary $SO(3)$ rotation. Using the fact that the eigenvalues of $\vec{v}\cdot\vec{\sigma}$ are $\pm \|\vec{v}\|$ we get 
\begin{equation}\hspace{1cm}
\| ({\cal R} \vec{\beta})_i \|_1 ^2=
\sum_j \mu_j^2 |{\cal R}_{ij}|^2\ ,
\end{equation}
so the maximum value of $\| ({\cal R} \vec{\beta})_1 \|_1 ^2 + \| ({\cal R}\vec{\beta})_{2} \|_1 ^2$ occurs for ${\cal R}_{13}={\cal R}_{23}=0$. Then
\begin{equation}\hspace{1cm}
\begin{aligned}
    \max_{{\cal R}}\{\| ({\cal R} \vec{\beta})_1 \|_1 ^2 + \| ({\cal R}\vec{\beta})_2 \|_1 ^2\}&=
\sum_{i=1,2}\sum_j \mu_j^2 |{\cal R}_{ij}|^2\\
    &=\mu_1^2 +\mu_2^2 
\end{aligned}
\end{equation}
(independent of ${\cal R}_{ij}$). Plugging this result in (\ref{generalresult}) we recover Eq.(\ref{CHSH-abC_max}).
\end{itemize}

\section{Applications}
\label{sec:applications}

\subsection{Lower and upper bounds on Bell violation}\label{sec:3}

From the general expression for the maximal Bell violation, Eq.(\ref{generalresult}), we can easily extract simple lower and upper bounds on $\mathcal{B}$.
They are useful as a quick exploration of nonlocality.

The lower bound comes from simply taking ${\cal R}=\mathbb{1}_3$. In other words, once the density matrix has been expressed as in Eq.(\ref{rho}), we can assure that
\begin{equation}
    \begin{aligned}
        \mathcal{B}  
        &  \, \geq  \,  2 
\sqrt{\| \beta_1 \|_1 ^2 + \| \beta_2 \|_1 ^2 } \\ 
  &  \, =  \,   2 
\left[
           \left(\sum_{i=1}^{d} |\lambda_i^{(1)}| \right) ^2  +
           \left( \sum_{i=1}^{d} |\lambda_i^{(2)}| \right) ^2
\right]^{1/2},
\label{lowerbound}
    \end{aligned}
\end{equation}
where in this case $\lambda_i^{(1,2)}$ stand for the eigenvalues of the initial $\beta_1, \beta_2$ matrices (no rotation applied). More precisely, $\beta_1, \beta_2$ correspond to the beta matrices with larger trace-norm.

We find that this bound is typically close to the actual value and becomes tighter as the state becomes more mixed. In Section \ref{sec:random}, we consider a random set of states in different dimensions. For $d=4$, the lower bound (\ref{lowerbound}) provides an estimate of $\mathcal{B} $ within a 10-15 $\% $ error. For the
d = 10 set, which typically has lower purity, the precision of the
approximation is around 5$\% $.

In order to get an upper bound on $\mathcal{B}$ from Eq.(\ref{generalresult}), we use the inequality $\norm{\,\cdot\,}_1^2\leq d\, \norm{\,\cdot\,}_2^2$ involving the 1 and 2-norm over $d\times d$ matrices~\footnote{The 2-norm of a squared matrix is defined by $\|M\|_2^2=\Tr{MM^\dagger}=\sum_i |\lambda_i|^2$, in analogy to the 2-norm of vectors.}, 
so that 
\begin{equation}
    \begin{aligned}
        & \sum_{a=1}^2 \| ({\cal R} \vec{\beta})_a \|_1^2   \leq   \sum_{a=1}^3 \| ({\cal R} \vec{\beta})_a \|_1^2  \\
        &\leq d\,\sum_{a=1}^3 \| ({\cal R} \vec{\beta})_a \|_2^2 \ = d\,\sum_{a=1}^3 \| \beta_a \|_2^2.
    \label{chain}
    \end{aligned}
\end{equation}
The equality comes from the fact that the last expression is invariant under $O(3)$ rotations, so we can take the initial $\beta$-matrices to evaluate the upper bound. Hence,
\begin{equation}\hspace{-0.2cm}
 \mathcal{B}\leq  
2 \sqrt{d\sum_{a=1}^3\| \beta_a \|_2 ^2}=2\sqrt{d}
\left[\sum_{a=1}^3
           \left(\sum_{i=1}^{d} |\lambda_i^{(a)}|^2 \right) 
\right]^{1/2}.
\end{equation}
An illustration of the goodness of this bound is provided in Section \ref{sec:5}, see Fig. \ref{fig:E1}.

Gathering the lower and upper bounds derived in this section, we present the following result.

\begin{proposition}
Under the same hypotheses as in Theorem \ref{Lagrange} the following inequalities are satisfied
\begin{equation}
2 
\sqrt{\sum_{a=1}^2 \| \beta_a \|_1 ^2}
    \leq\mathcal{B}\leq  
2\sqrt{d\sum_{a=1}^3\| \beta_a \|_2 ^2}\, .
\label{bounds}
\end{equation}
\end{proposition}

\vspace{0.3cm}

\subsection{Purity and non-locality} \label{sec:purity}

Purity is simply defined as $\Tr \rho^2$,
and provides a simple measure of how pure a state is \cite{Benenti}, in particular $\Tr{\rho^2}=1$ ($\Tr{\rho^2}=1/D$) for a pure (completely mixed) state, where $D$ is the dimension of the Hilbert space. 
Clearly, a maximally-mixed bipartite state is separable (and thus local in the Bell sense), since the density matrix is proportional to the identity. In contrast, a pure state, if entangled, is non-local as well \cite{GISIN1991201,Popescu:2002won}. Then, one can wonder what is the minimum degree of purity a state should have in order to present non-local correlations.

This question can be adressed for qubit-qudit systems, using our previous results.
From the expression of $\rho$ in Eq.(\ref{rho}) it follows that 
\begin{equation}
\Tr{\rho^2}=\frac{1}{2}\left[\Tr{\beta_0^2}\ + \sum_{i=1,2,3}\Tr{\beta_i^2}\right].
\end{equation}
Here $\beta_0=\rho_B=\Tr_A\{\rho\}$ is the reduced density matrix in Bob's side, note that $\Tr{\rho_B^2}\geq1/d$. On the other hand, $\sum_{i=1,2,3}\Tr{\beta_i^2}= \,\sum_{i=1}^3 \| \beta_i \|_2^2$. Then, using relations (\ref{chain}) and (\ref{generalresult}) we get
\begin{equation}
    \begin{aligned}
        \Tr{\rho^2}&\geq  \frac{1}{2}\left[\Tr{\rho_B^2} + \frac{1}{4d}\mathcal{B}^2\right]\\[2mm]
        & \geq  \frac{1}{2d}\left[1 + \frac{1}{4}\mathcal{B}^2\right]\ ,
    \end{aligned}
\end{equation}
and thus a bound on the purity of $\rho$ in order to have CHSH violation ($\mathcal{B}>2$):
\begin{equation}\hspace{-0.25cm}
\text{Bell-Violation}\ \Rightarrow\ \Tr{\rho^2}>\frac{1}{2}\left[\Tr{\rho_B^2}+\frac{1}{d}\right] \geq \frac{1}{d}\ .
\end{equation}
Notice here that generically $\Tr{\rho^2}\geq1/(2d)$.

\subsection{Non-locality of $2\times d$ Random Matrices}
\label{sec:random}

As it is well known, for a bipartite quantum system most pure states are entangled and thus violate Bell-inequalities. The situation is more complicated to analyze (and actually unknown) for general mixed states. 

Our main result on CHSH-violation, Eq.(\ref{generalresult}),
allows to explore the statistics of non-locality in different sets of (generically mixed) qubit-qudit states, in particular in random mixed states. As a matter of fact, there is not a unique way to define a ``random mixed state", since it depends on the the probability measure used (the following results are however similar in all cases).
A well motivated one is the Random  Bures set. It can be generated from the Ginibre ensemble and the random unitary matrix
ensemble -- distributed according to the Haar measure --.    Let $G$ and $U$ be two $ D \times D$ random matrices of these ensembles. Then
\begin{equation}
\rho_{\rm Bures} = \frac{ (\mathds{1} + U) G G^\dagger  (\mathds{1} + U^\dagger)}{\Tr{(\mathds{1} + U) G G^\dagger  (\mathds{1} + U^\dagger)}}
\end{equation}
provides a random $ D \times D$ Bures state \cite{Al_Osipov_2010}.

In order to study the statistics of CHSH-nonlocality of these states for qubit-qudit (where $D=2d$)
we have generated sets of $\sim 10,000$ density matrices for different $d$ values, and evaluated the optimal Bell operator value for each state. 
The corresponding probability distributions of $\mathcal{B}$ are closed to normal, as displayed in Fig 1. Indeed, the value $\langle {\cal O}_{\rm Bell} \rangle=2$ is always achievable by allowing all the observables to be fully degenerate, i.e. trivial (which we did not for the Alice's ones). Hence, the really meaningful statistics corresponds to the nonlocal region, $\mathcal{B}>2$.
Clearly, the higher $d$, the less likely for a mixed state to present nonlocality.
More precisely, for $d=2, 4$ the $p-$value of non-locality is $9.1\%, 1.4\%$ respectively. For $d\geq 10$, such probability drops below $10^{-4}$.
\begin{figure}[ht!]
   \begin{center}
  \includegraphics[scale=0.65]{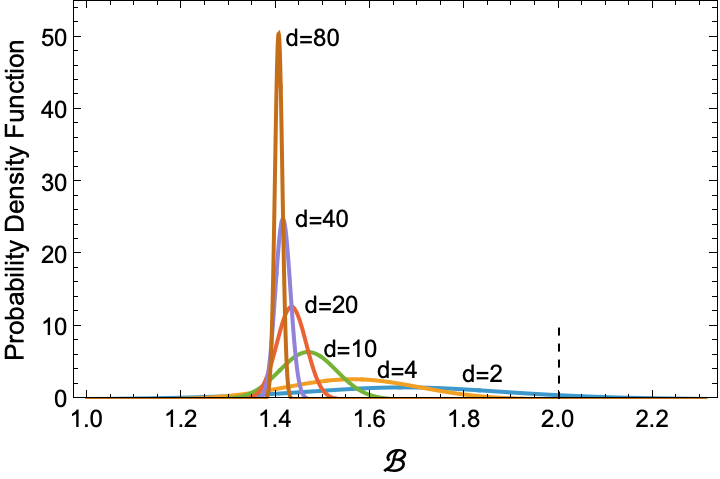}
     \end{center}
	\caption{From right to left, probability distribution of $\mathcal{B}$
	 for $d= 2, 4, 10,  20, 40, 80$. 
   }
    \label{fig:px}
\end{figure}

\subsection{Embeddings}\label{sec:4}
Having a recipe for the optimal CHSH-violation for $2 \times d$ systems allows us to address the following question: if we embed Bob's Hilbert space in one of  larger dimension, is it possible to improve the amount of Bell-violation by a suitable choice of the new (higher dimensional) $\tilde B, \tilde B'$ observables? As we are about to see, the answer is negative.
Let us start with a generic  state in a ${\cal H}_2 \otimes {\cal H}_{d_1}$ Hilbert space, characterized by a density matrix
\begin{equation}
    \rho = \dfrac{1}{2}\left(\openone_2 \otimes \beta_0+\sigma_1 \otimes \beta_1+\sigma_2 \otimes \beta_2+\sigma_3 \otimes \beta_3 \right),
\end{equation}
where $\{\beta_0,\beta_i\}$ are $d_1\times d_1$ matrices. Now let us consider Bob's Hilbert space as part of a higher dimensional one, ${\cal H}_{d_1}\subset{\cal H}_{d_2}$ with
$d_2>d_1$. Thus we embed the above state in the new Hilbert space by considering the $\{\beta_0,\beta_i\}$ matrices as the upper-left block of a block diagonal $d_2\times d_2$ matrix:
\begin{equation}
\begin{aligned}
    \beta_0\rightarrow\tilde\beta_0 
    & = 
    \left(\begin{array}{cc}
      \beta_0   & \mathbb{0}_{d_1\times (d_2-d_1)} \\
      \mathbb{0}_{(d_2-d_1)\times d_1} & \mathbb{0}_{(d_2-d_1)\times (d_2-d_1)} 
    \end{array}\right), \\[4mm]    
\beta_i\rightarrow\tilde\beta_i & =  \left(\begin{array}{cc}
      \beta_i   & \mathbb{0}_{d_1\times (d_2-d_1)} \\
      \mathbb{0}_{(d_2-d_1)\times d_1} & \mathbb{0}_{(d_2-d_1)\times (d_2-d_1)} 
    \end{array}\right),
    \end{aligned}
    \label{betas}
\end{equation}
where the $\mathbb{0}$ matrices have all entries vanishing. In terms of the higher-dimension observables and $\beta-$matrices, Eq.(\ref{r_squared2}) reads:
\begin{equation}
\mathcal{B} =2 \max_{B,B', {\cal R}} 
 \sqrt{ \left| \Tr{({\cal R}  \vec{\beta})_1 \cdot B} \right |^2  +  \left| \Tr{({\cal R}  \vec{\beta})_2 \cdot B'} \right |^2}.
 \label{r_squared3}
\end{equation}
Note that the ${\cal R}  \vec{\tilde\beta}$ matrices are block-diagonal, with the same texture of zeroes as matrices (\ref{betas}). Hence, they have the same $d_1$ eigenvalues as ${\cal R}  \vec{\beta}$ plus $d_2-d_1$ zeroes. Therefore, for a given rotation ${\cal R}$, the {\em optimal} choice for $\tilde B, \tilde B'$ in Eq.(\ref{r_squared3}) yields the same result as the optimal choice in the $ {\cal H}_2 \otimes {\cal H}_{d_1}$ system.

\subsection{A qubit-qutrit case study}\label{sec:5}

To illustrate the use of the general result on the maximal CHSH-violation (\ref{generalresult}) let us consider an example in the context of the qubit-qutrit system. As it is well known, for mixed states entanglement does not necessarily leads to violation of quantum realism (i.e. Bell-violation). A popular example of this fact in the qubit-qubit case are the Werner states, $\rho=\frac{1}{4}(
\mathbb{1}_2\otimes \mathbb{1}_2 - \eta \sum_{i} \sigma_i\otimes \sigma_i)$, which for $1/3<\eta\leq 1/\sqrt{2}$ are entangled but do not violate any Bell inequality. For a qubit-qutrit system we can perform a similar analysis for the CHSH scenario, using both our result (\ref{generalresult}) and the fact that in this case the Peres-Horodecki \cite{PhysRevLett.77.1413,Horodecki:1996nc} criterion, i.e. the existence of some negative eigenvalue of the  partially transposed matrix $\rho^{T_2}$,
provides a necessary and sufficient condition for entanglement. For the sake of concreteness, let us consider the qubit-qutrit state
\begin{equation}
\rho=x|\psi_1\rangle\langle\psi_1|
 \ + \  y|\psi_2\rangle\langle\psi_2| \ +\  z|\psi_3\rangle\langle\psi_3|,
\end{equation}
where $0\leq (x,y,z)\leq 1$ with $x+y+z=1$, and (in an obvious notation)
\begin{eqnarray}
&& |\psi_1\rangle=\frac{1}{\sqrt{2}}\left(|00\rangle+|11\rangle\right),\ \
|\psi_2\rangle=\frac{1}{\sqrt{2}}\left(|01\rangle+|12\rangle\right), \nonumber \\
&& |\psi_3\rangle=\frac{1}{\sqrt{2}}\left(|02\rangle+|10\rangle\right).
\end{eqnarray}
Explicitly,
\begin{equation}
\rho = \frac{1}{2} \left( \! \begin{array}{cccccc}
x & 0 & 0 & 0 & x & 0 \\
0 & y & 0 & 0 & 0 & y \\
0 & 0 & 1-x-y &1-x-y & 0 & 0  \\
0 & 0 & 1-x-y & 1-x-y & 0 & 0 \\
x & 0 & 0 & 0 & x & 0 \\
0 & y & 0 & 0 & 0 & y 
\end{array} \! \right) \,.
\label{rho 2x3}
\end{equation}
The physical region, where $\rho$ is positive definite, corresponds to $x+y\leq 1$ (triangle in Fig.\ref{fig:E1}).
Using the Peres-Horodecki criterion, it is easy to check that $\rho$ is entangled for any value of $x,y$, except for $x=y=1/3$. 
Fig.\ref{fig:E1}, left panel, shows the value of the logarithmic negativity $E=\log_2\left(\|\rho^{T_2}\|_1\right)$ in the $x-y$ plane. The logarithmic negativity, which provides a 
sound measurement of entanglement \cite{Vidal:2002zz,
Plenio:2005cwa}, is greater than 0 in the whole physical region except at that particular point. 

\begin{figure}[ht!] 
\begin{center}
\includegraphics[width=7.8cm,clip=]{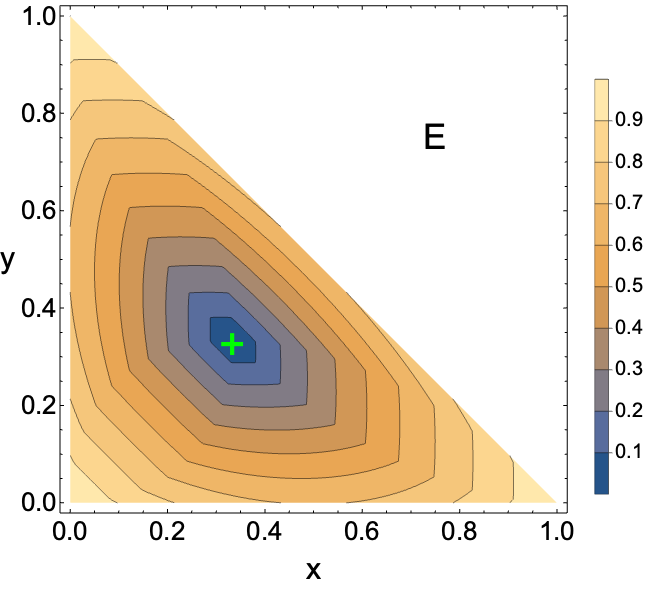}
\includegraphics[width=7.8cm,clip=]{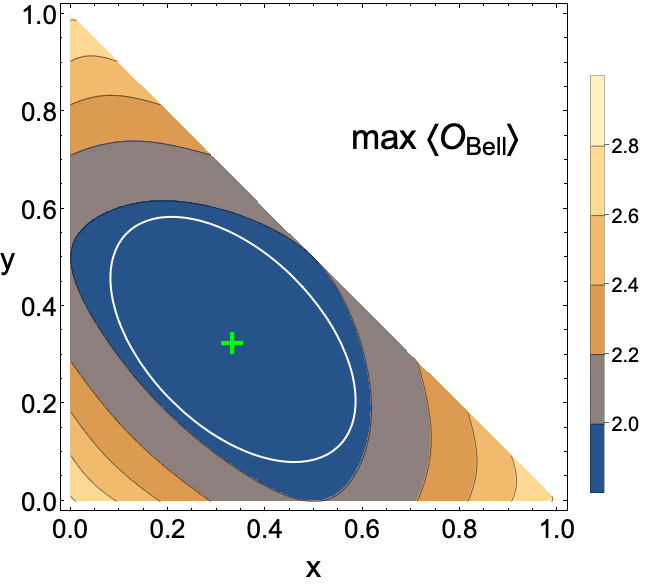}
\caption{Values of the logarithmic negativity. $E=\log_2\left(\|\rho^{T_2}\|_1\right)$ (upper panel) and 
$\mathcal{B}$ (lower panel) for the qubit-qutrit model described by the density matrix of Eq.(\ref{rho 2x3}). The model is entangled ($E>0$) in the whole physical region, except for $x=y=1/3$, while it violates local realism for $\mathcal{B}>2$. The area enclosed by the white line corresponds to the region where CHSH violation is ruled out by the upper bound in Eq.(\ref{bounds}).} 
\label{fig:E1}
\end{center}
\end{figure} 

For the analysis of the CHSH-violation, we first express $\rho$ in the form (\ref{rho}), which amounts to the following $\beta-$matrices:
\begin{equation}
\begin{aligned}
\beta_0 &= \frac{1}{2}
\left(
\begin{array}{ccc}
 1-y&0&0\\
 0&x+y&0\\ 
 0&0&1-x 
 \end{array}
\right), \\[2mm]
\beta_1 & =\frac{1}{2}
\left(
\begin{array}{ccc}
0&x&-x-y+1\\
x&0&y\\ 
-x-y+1&y&0 
\end{array}
\right),  \\[2mm]
\beta_2 & = \frac{1}{2}
\left(
\begin{array}{ccc}
0&i x&i (x+y-1) \\
-i x&0&i y\\
-i (x+y-1)&-i y&0
\end{array}
\right),\\[2mm]
\beta_3 &=\frac{1}{2}
\left(
\begin{array}{ccc}
2 x+y-1&0&0\\0&y-x&0\\
0&0&-x-2 y+1
\end{array}
\right).
\label{betas23}
\end{aligned}
\end{equation}
Plugging these expressions in
Eq.(\ref{generalresult}) and performing a simple numerical 
optimization we can obtain the maximal Bell-violation  in the $x-y$ plane, which is shown in Fig.\ref{fig:E1}, 
right panel. Similarly to the Werner qubit-qubit states, there is a sizeable region in which the state is entangled but $\mathcal{B}\leq 2$.

\section{Summary and conclusions}
\label{sec:6}
We have considered the violation of CHSH-like inequalities in the context of a qubit-qudit system with arbitrary dimension. These inequalities represent a crucial test of local realism, i.e the possibility that the outcomes of physical measures on the system could be reproduced by a (classical) theory of hidden variables. The violation of such inequalities requires that the state is entangled, but (for mixed states) the opposite is not necessarily true. In previous literature \cite{Pironio_2014} it was shown that for these systems in a bipartite scenario in which one party has binary inputs and outcomes, the ``classical" polytope, i.e. the region of probability distribution of  observables $A, A', B, B'$ which is compatible with local realism, is bounded by CHSH-type \cite{Clauser:1969ny} inequalities. However, given a state $\rho$, this does not solve the problem of determining the maximum CHSH violation and thus whether the system can be described by a classical (local-realistic) theory. The usual recipes for a qubit-qubit system \cite{HORODECKI1995340} cannot be applied beyond the lowest dimensionality. Hence, in principle one should explore all possibilities for the $A, A', B, B'$ observables involved in a CHSH inequality, an expensive computational task, which entails to optimize $\sim 2d^2$ parameters and thus increases quickly with the dimension of the qudit.

In this paper we have addressed the task of evaluating the maximal CHSH violation for a generic $\AB{given}$ qubit-qudit state, obtaining easily computable expressions. Our central result, given in Eq.(\ref{generalresult}), generically amounts to a simple optimization in three angles, independently of the qudit dimension. Moreover, once the maximum value is obtained, it automatically allows to construct the observables $A,A',B,B'$ for which the optimal CHSH value is attained.
We also give lower (and usually close to optimal) bounds and upper bounds on the Bell-violation, which can be immediately computed. We have applied our general results to address a number of issues. Namely, we have obtained a bound on the degree of purity required in a system to exhibit nonlocality. Also, we have studied the statistics of nonlocality in random density matrices.
Besides, we have shown that it is not possible to improve the amount of CHSH-violation by embedding Bob's Hilbert space in one of larger dimension. 
Finally, as an example of the use of our results we have considered a 2-parameter family of density matrices in the context of a qubit-qutrit system and determined the region of such parameter space in which the state is entangled and the region where local realism is violated, showing that both are correlated but the former is broader than the latter.

The results presented here can be used for any qubit-qudit system, independently of its physical nature; e.g. in the analysis of non-local correlations in top-$W$ \cite{Aguilar-Saavedra:2023hss,Subba:2024mnl} or photon-$Z$ production \cite{Morales:2023gow, Morales:2024jhj} at the LHC,
electron-nuclear spin qudit systems \cite{Macedonio:2025omv}
or even (in the large-$d$ limit) hybrid discrete-continuous systems such as a cavity atom-light system\cite{Halder:2023iod,Bernal:2024ege}. 

Furthermore, it would be interesting to investigate whether a theorem similar to Theorem \ref{Lagrange} can be formulated for other types of (2-2-2) inequalities, such as the family of Tilted CHSH Inequalities \cite{PhysRevLett.108.100402} (which includes the CHSH inequality as a special case). Concerning this point, it was shown in \cite{PhysRevLett.132.250205} that the maximum quantum value (Tsirelson's bound \cite{Cirelson:1980ry}) for this family is asymptotically inequivalent to that attained via Hardy's nonlocality \cite{PhysRevLett.68.2981}. This result was obtained using the explicit quantum realization that reaches the Tsirelson bound, namely the maximally entangled state and the corresponding observables for qubit-qubit states.
However, for qubit-qudit  it is possible to hit the quantum bound not only with the
Bell state but also with mixed states 
\cite{PhysRevLett.68.3259}. Then Theorem \ref{Lagrange} could help to characterize all such mixed states.  Hence, following the lines  of ref.~\cite{PhysRevLett.132.250205}, it would be of high interest to check whether these new higher-dimensional  realizations of the quantum bound are still asymptotically inequivalent to those of Hardy's nonlocality.

\section*{Acknowledgements}

	We are grateful to  J.A. Aguilar-Saavedra for useful discussions. The authors acknowledge the support of the Agencia Estatal de Investigacion through the grants ``IFT Centro de Excelencia Severo Ochoa CEX2020-001007-S" and PID2022-142545NB-C22 funded by MCIN/AEI/10.13039/501100011033 and by ERDF. The work of A.B. is supported through the FPI grant PRE2020-095867 funded by MCIN/AEI/10.13039/501100011033. \\

{\em Data availability:} The data that support the findings of this article are not publicly available. The data are available from the authors upon reasonable request.

\appendix
\section*{Appendix A}
\label{POVMAppend}
In this Appendix we show the equivalence between the CHSH inequality in term of probabilities with that of the CHSH inequality as a bound on the expectation value of an Hermitian operator.

The CHSH inequality as a function of the probabilities of getting outcomes $a,b=\pm1$ given inputs $x,y=0,1$ for Alice and Bob reads
\begin{equation}\label{CHSHExp}
    E (0,0)+ E (0,1)+E (1,0)-E (1,1) \leq 2,
\end{equation}
where we denote $E(x,y)= P(1,1|x,y)-P(1,-1|x,y)-P(-1,1|x,y)+P(-1,-1|x,y)$ and $P(a,b|x,y)$ is the probability of getting outcomes $a,b=\pm1$ given inputs $x,y=0,1$. Each probability is computed by the Born rule via
\begin{equation}\label{BornRule}
   P(a,b|x,y)=\Tr{\rho\left( F_{a|x}^A\otimes F_{b|y}^B\right)},
\end{equation}
where $F_{a|x}^A$ and $ F_{b|y}^B$ are respectively the positive semi-definite operators for both Alice's and Bob's POVMs (not necessarily PVMs) associated with outcomes $a,b$ and inputs $x,y$. By definition of POVM, the following identities hold:
\begin{equation}\label{identity}  
    F_{1|x}^A+F_{-1|x}^A=\openone, \quad F_{1|y}^B+F_{-1|y}^B=\openone,
\end{equation}
with $x,y=0,1$. Replacing the expression of the probabilities \eqref{BornRule} into Eq.~\eqref{CHSHExp} yields 
\begin{equation}
\expval{{\cal O}_{\rm Bell}}=\expval{A \otimes (B+B') + A'\otimes (B-B')}\leq 2,
    \label{CHSHPOVMs2}
\end{equation}
 with $\langle A\otimes B\rangle=\Tr{\rho \left(A\otimes B\right)}$ and
\begin{equation}
\begin{aligned}
     A=F_{1|0}^A-F_{-1|0}^A=\openone-2F_{-1|0}^A,\\
     B=F_{1|0}^B-F_{-1|0}^B=\openone-2F_{-1|0}^B, 
\end{aligned}
\end{equation}
where in the second step of each equality we have used (\ref{identity}).
Analogous relations define $A'$ and $B'$ for $x,y=1$. Since all $F_{a|x}^A$ and $F_{b|y}^B$ are positive semi-definite operators satisfying (\ref{identity}),
their eigenvalues are all non-negative and bounded from above by 1. In consequence, from the definition of the operators $A$ and $B$ (and, in the same fashion, for $A'$ and $B'$), we deduce that their eigenvalues have to be contained in the $[-1,1]$ interval. For instance, when dealing with PVMs the $F$ operators are projectors and the eigenvalues of $A,A',B,B'$ are exactly $\pm 1$.

\bibliographystyle{style2.bst}  

\bibliography{references}

\end{document}